\documentclass[prl, twocolumn, aps, preprintnumbers,showpacs, superscriptaddress]{revtex4}
\usepackage{latexsym}
\usepackage{amsmath}
\usepackage{amssymb}
\usepackage{amsfonts}
\usepackage{amsthm}
\usepackage{ifthen}
\usepackage{revsymb}
\usepackage{graphicx}
\usepackage{bbm}
\usepackage{ifpdf}
\usepackage{hyperref}

\newtheorem{theorem}{Theorem} 

\newtheorem*{lemma}{Lemma}
\newtheorem*{definition}{Definition}

\newcommand{\sket}[1]{{\ensuremath{\lvert#1\rangle}}}
\newcommand{\lket}[1]{{\ensuremath{\left\lvert#1\right\rangle}}}
\newcommand{\ket}[1]{\if@display\lket{#1}\else\sket{#1}\fi}

\newcommand{\sbra}[1]{{\ensuremath{\langle#1\rvert}}}
\newcommand{\lbra}[1]{{\ensuremath{\left\langle#1\right\rvert}}}
\newcommand{\bra}[1]{\if@display\lbra{#1}\else\sbra{#1}\fi}

\newcommand{\sbraket}[2]{{\ensuremath{\langle#1\rvert#2\rangle}}}
\newcommand{\lbraket}[2]{{\ensuremath{\left\langle#1\!\left\rvert\vphantom{#1}#2\right.\!\right\rangle}}}
\newcommand{\braket}[2]{\if@display\lbraket{#1}{#2}\else\sbraket{#1}{#2}\fi}

\newcommand{\sketbra}[2]{{\ensuremath{\lvert #1\rangle\!\langle #2\rvert}}}
\newcommand{\lketbra}[2]{{\ensuremath{\left\lvert #1\right\rangle\!\!\left\langle #2\right\rvert}}}
\newcommand{\ketbra}[2]{\if@display\lketbra{#1}{#2}\else\sketbra{#1}{#2}\fi}

\newcommand{\proj}[1]{\ketbra{#1}{#1}}

\newcommand*{\id}{\mathrm{id}}

\DeclareMathOperator{\tr}{tr}
\newcommand{\strace}[2][@]{\ensuremath{\tr\ifthenelse{\equal{#1}{@}}{}{_{#1}}(#2)}}
\newcommand{\ltrace}[2][@]{\ensuremath{\tr\ifthenelse{\equal{#1}{@}}{}{_{#1}}\left(#2\right)}}
\newcommand{\trace}[2][@]{\if@display\ltrace[#1]{#2}\else\strace[#1]{#2}\fi}

\newcommand{\sspan}[1]{{\ensuremath{\operatorname{span}(#1)}}}
\newcommand{\lspan}[1]{{\ensuremath{\operatorname{span}\left(#1\right)\!}}}
\newcommand{\vspan}[1]{\if@display\lspan{#1}\else\sspan{#1}\fi}

\newcommand*{\assign}{\ensuremath{\kern.5ex\raisebox{.1ex}{\mbox{\rm:}}\kern -.3em =}}

\newcommand{\eps}{\varepsilon}

\newcommand*{\mF}{\mathcal{F}}
\newcommand*{\mFaug}{\mathcal{F}_{\!\mathit{aug}}}

\newcommand{\qp}{\pi}

\newcommand{\epsclose}[1][\eps]{\approx_{#1}}

\newcommand{\A}{{\sf A}}
\newcommand{\dA}{{\sf A}'}
\newcommand{\B}{{\sf B}}
\newcommand{\dB}{{\sf B}'}
\newcommand{\hA}{\hat{\sf A}}
\newcommand{\dhA}{\hat{\sf A}'}

\newcommand{\dhB}{\hat{\sf B}'}

\newcommand{\hut}[1]{#1}

\newcommand{\refe}{R}
\newcommand{\Ain}{U}
\newcommand{\Bin}{V}
\newcommand{\FAin}{\tilde{U}}
\newcommand{\FBin}{\tilde{V}}
\newcommand{\FAout}{\tilde{X}}
\newcommand{\FBout}{\tilde{Y}}
\newcommand{\Bcheatout}{Y'}
\newcommand{\Aout}{X}
\newcommand{\Bout}{Y}
\newcommand{\aux}{K}
\newcommand{\Apur}{X'_1}
\newcommand{\Bpur}{Y'_1}
\newcommand{\pur}{P}
 
\newcommand*{\cD}{{\cal D}} 

\newcommand*{\cW}{{\cal W}} 

\newcommand*{\cE}{{\cal E}}

\begin{document}

\title{Complete Insecurity of Quantum Protocols for Classical Two-Party Computation}

\author{Harry \surname{Buhrman}}
\affiliation{University of Amsterdam and CWI Amsterdam, The Netherlands}

\author{Matthias \surname{Christandl}}
\affiliation{Institute for Theoretical Physics, ETH Zurich, Wolfgang-Pauli-Strasse 27, CH-8093 Zurich, Switzerland}
\author{Christian \surname{Schaffner}}
\affiliation{University of Amsterdam and CWI Amsterdam, The Netherlands}


\date{\today}
\begin{abstract}
A fundamental task in modern cryptography is the joint computation
of a function which has two inputs, one from
Alice and one from Bob, such that neither of the two can learn more
about the other's input than what is implied by the value of the
function. In this Letter, we show that any
quantum protocol for the computation of a classical deterministic function that outputs the result to both parties (two-sided
computation) and that is secure against a cheating Bob can be
completely broken by a cheating Alice. Whereas it is known that
quantum protocols for this task cannot be completely secure, our
result implies that security for one party implies complete insecurity for the other. Our
findings stand in stark contrast to recent protocols for weak coin tossing,
and highlight the limits of cryptography within quantum mechanics.  
We remark that our conclusions remain valid, even if security is only required to be approximate and if the function that is computed for Bob is different from that of Alice.
\end{abstract}
\maketitle

 Traditionally, cryptography has been understood as the art of ``secret
 writing``, i.e., of sending messages securely from one party to
 another. Today, the research field of cryptography comprises much more
 than encryption and studies all aspects of secure communication and
 computation among players that do not trust each other, including tasks such as electronic voting and auctioning. Following the excitement that the exchange of quantum particles may allow for the distribution of a key that is unconditionally secure~\cite{BB84,E91}, a level of security
 unattainable by classical means, the question arose whether other fundamental cryptographic tasks could be
 implemented with the same level of security using quantum
 mechanical effects. For oblivious transfer and bit commitment, it was shown that the answer is negative~\cite{LoChau96BCprl, Mayers97BC}. Interestingly, however,
 a weak version of a coin toss can be implemented by quantum
 mechanical means~\cite{Mochon07}.


In this Letter we study the task of secure two-party
computation. Here, two mistrustful players, Alice and Bob, wish to
compute the value of a classical deterministic function $f$, which takes an input
$u$ from Alice and $v$ from Bob, in such a way that both learn the
result of the computation and that
none of the parties can learn more about the other's input, even by deviating from the protocol. As our main result we show that any quantum protocol which is secure against a cheating Bob can be completely broken by a cheating Alice. Formally, we design an attack by Alice which allows her to compute the value of the function $f$ for all of her inputs (rather than only a single one, which would be required from a secure protocol). 

Our result strengthens the impossibility result for two-sided secure
two-party computation by Colbeck, where he showed that Alice can
always obtain more information about Bob's input than what is implied
by the value of the function~\cite{Colbeck07}. In a similar way, we
complement a result by Salvail, Schaffner and Sot\'akov\'a~\cite{SSS09}
showing that any quantum protocol for a non-trivial primitive
necessarily leaks information to a dishonest player. Our result is
motivated by Lo's impossibility result for the case where only Alice
obtains the result of the function (one-sided
computation)~\cite{Lo97}. Lo's approach is based on the idea that Bob
does not have any output; hence, his quantum state cannot depend on
Alice's input. Then, Bob has learned nothing about Alice's input and a
cheating Alice can therefore still change her input value (by
purifying the protocol) and thus cheat. 

In the two-sided case, this
approach to proving the insecurity of two-party computation fails as
Bob knows the value of the function and has thus some information
about Alice's input. In order to overcome this problem we develop a
new approach. We start with a formal definition of security 
based on the standard real/ideal-world paradigm from modern cryptography. 
In our case of a classical functionality, this definition guarantees the existence of a classical
input for Bob in the ideal world, even if he is, in the real world, dishonestly purifying his steps of the protocol. Since real and ideal are indistinguishable for a secure protocol and since a purification of the classical input cannot be part of Bob's systems, Alice can now obtain a copy of this input by applying a unitary---constructed with help of Uhlmann's theorem---to her output registers and, henceforth, break the protocol.

We wish to emphasize that the above
conclusion remains valid if the protocol is not required to be
perfectly secure (nor perfectly correct). More precisely, if the
protocol is secure up to a small error against cheating Bob, then
Alice is able to compute the value of the function for all of her
inputs with only a small error. Since the error
is independent of the number of inputs that both Alice and Bob
have, our analysis improves over Lo's result in the
one-sided case. In fact, our results apply to this case since, more generally, they remain true should Bob receive the output of a function $g$, different from Alice's $f$, as a careful look at our argument reveals.

\medskip\noindent {\bf Security Definition.}  Alice and Bob, at distant locations and only connected with a quantum channel, wish to execute a protocol that takes an input $u$ from Alice and an input $v$ from Bob and that outputs the value $f(u, v)$ of a classical deterministic function $f$ to both of them. Since Alice
does not trust Bob, she wants to be sure that the protocol does not
allow him to extract more information about her input than what is
implied by the output value of the function. The same should be true
if Alice is cheating and Bob is honest. Whereas for simple functions this intuitive notion
of security can be made precise by stating a list of security
requirements for certain quantum states of Alice and Bob, such an approach seems
very complicated and prone to pitfalls for general functions $f$, in
particular, if we want to consider protocols that are only
approximately secure. We therefore follow the modern literature on cryptography where such situations have been in
the center of attention for many years (cf.~zero-knowledge,
composability) and where a suitable notion of security, known as the
real\slash ideal-world paradigm, has been firmly established. 

In this paradigm we first define an ideal situation in
which everything is computed perfectly and securely and call this the
\emph{ideal functionality}. Informally, a
two-party protocol is secure if it looks to the outside world just like the ideal functionality it is supposed to implement. More concretely, a protocol is deemed secure if for every adversarial strategy, or \emph{real
  adversary}, there exists an \emph{ideal adversary}
interacting only with the ideal functionality such that the execution of the protocol in the real world is \emph{indistinguishable} from this ideal world. If such a security guarantee holds, it is clear that a
secure protocol can be treated as a call to the ideal functionality
and hence, it is possible to construct and prove secure more
complicated protocols in a modular fashion. See~\cite{Canetti00, Canetti96,Goldreich04} and~\cite{Unruh2004, unruh-2009, BenOrMayers2004, FS09} for further information about this concept of security in the context of classical and quantum protocols, respectively.

There exist different meaningful ways to make the above informal
notion of the real\slash ideal-world paradigm precise. All these
notions have in common that the execution of the protocol by the
honest and dishonest players is modeled by a completely positive trace-preserving (CPTP)
map. Likewise, every ideal adversary interacting with the ideal
functionality is composed out
of CPTP maps modeling the pre- and postprocessing of the in- and
outputs to the ideal functionality (which is a CPTP map itself). A desirable notion of security is the following: for every real adversary there exists an ideal adversary,
such that the corresponding CPTP maps are (approximately)
indistinguishable. The natural measure of distinguishability of CPTP
maps in this context is the diamond norm, since it can be viewed as the
maximal bias of distinguishing real and ideal world by supplying
inputs to the CPTP maps and attempting to distinguish the outputs by
measurements (i.e.\ by interacting with an environment). This rather strong notion of security  naturally embeds into a composable framework for security in which also quantum key distribution can be proven secure (see e.g.~\cite{ChKoRe09}).

Since our goal is the establishment of a no-go theorem, we consider a
notion of security which is weaker than the above in two
respects. First, we do not allow the environment to supply an arbitrary
input state but only the purification of a classical input (see
definition of $\rho_{UVR}$ below), and second, we consider a different
order of quantifiers: instead of ``$\forall$ real adversary $\exists$
ideal adversary $\forall$ input, the output states are
indistinguishable'' as a security requirement we only require ``$\forall$ real adversary $\forall$
input $\exists$ ideal adversary, the outputs
states are indistinguishable.'' This notion of security is closely
related to notions of security considered in~\cite{FS09, unruh-2009}
and is further discussed in the appendix.

We will now give a formal definition of security. Following the
notation of~\cite{FS09}, we denote by $\A$ and $\B$ the real honest
Alice and Bob and add a prime to denote dishonest players $\dA,\dB$
and a hat for the ideal versions $\hat{\A},\hat{\B}$. The CPTP map corresponding to the protocol for
honest Alice and dishonest Bob is denoted by $\pi_{\A,\B'}$. Both
honest and dishonest players obtain an input, in Alice's case $u$ (in
register $\Ain$) and in Bob's case $v$ (in register $V$) drawn from
the joint distribution $p(u,v)$.  The output state of the protocol,
augmented by the reference $\refe$, takes the form $ \id_{\refe}
\otimes \qp_{\A,\dB}(\rho_{\Ain \Bin \refe}),$ where $\rho_{\Ain \Bin
  \refe}$ is a purification of $\sum_{u,v} p(u,v) \proj{u}_{\Ain} \proj{v}_{\Bin} $.

\begin{figure}
\includegraphics[trim = 10mm 50mm 40mm 70mm,width=1\columnwidth]{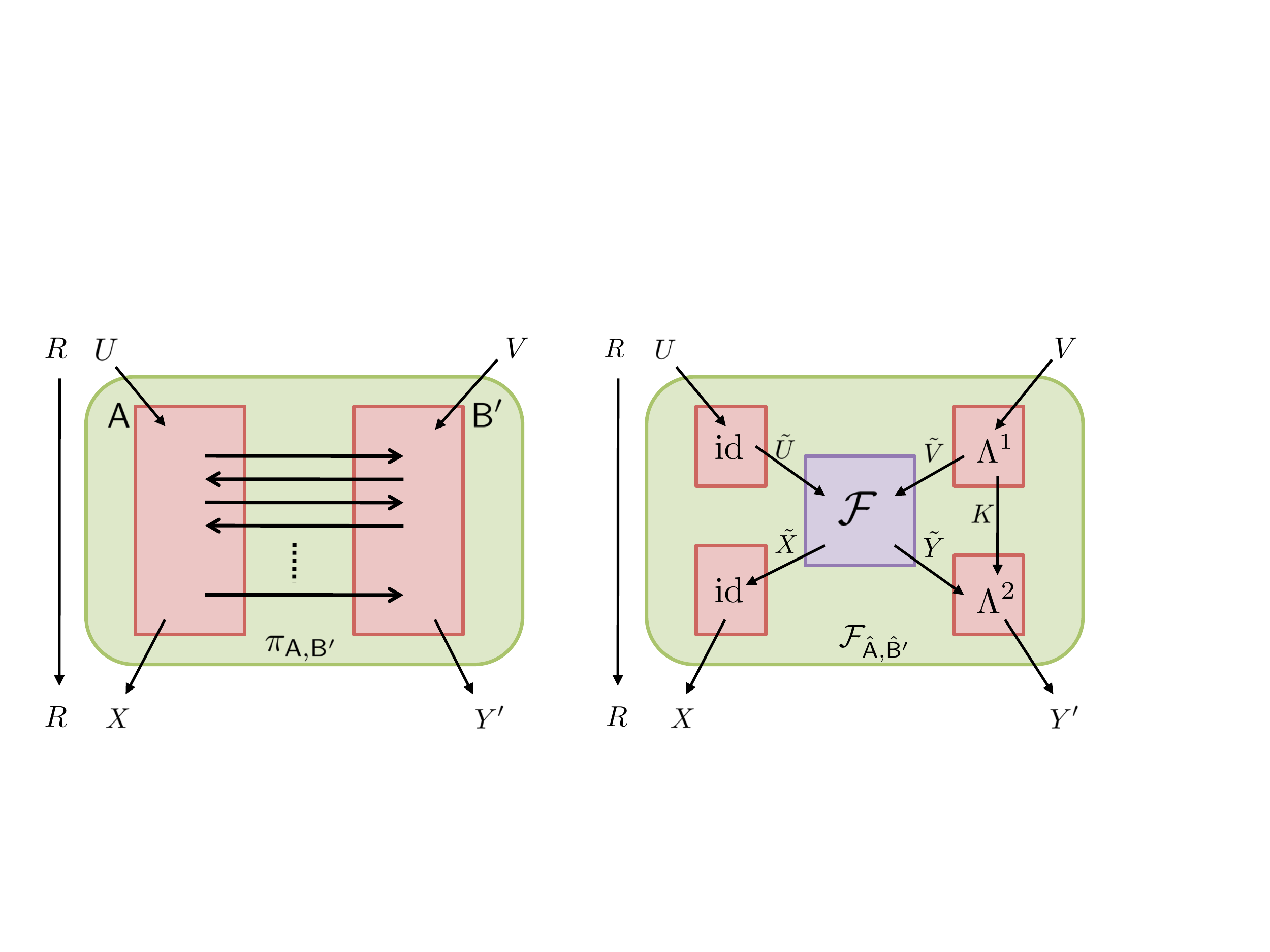}
\caption{Illustration of the security definition.
A protocol is secure against Bob, if the \emph{real protocol} (left) can
be simulated as an interaction with the \emph{ideal functionality $\mathcal{F}$}
(right). \label{fig:security}}
\end{figure}

Since we are faced with the task of the secure evaluation of a
\emph{classical} deterministic function, we consider an ideal
functionality $\mF$ which measures the inputs in registers $\FAin$ and $\FBin$ and outputs orthogonal states in registers $\FAout$ and $\FBout$ that correspond to the function
values. Formally, 
$\mF(\ketbra{u}{u'}_{\FAin}  \ketbra{v}{v'}_{\FBin}) :=\delta_{u,u'} \delta_{v,v'}   \proj{f(u, v)}_{\FAout} \proj{f(u, v)}_{\FBout},$
where $\delta$ denotes the Kronecker delta function.
When an ideal honest $\hA$ and an ideal adversary $\dhB$ interact with the ideal functionality, we denote the joint map by $\mF_{\hat\A,\hat\B'}: \Ain \Bin \rightarrow \Aout \Bcheatout$ (see Figure~\ref{fig:security}). $\hA$ just forwards the in- and outputs to and from the functionality, whereas $\dhB$ pre- and postprocesses them with CPTP maps $ \Lambda^1_{\Bin \rightarrow \FBin \aux}$ and $ \Lambda^2_{\aux \FBout \rightarrow \Bcheatout}$ resulting in a joint map
$ \mF_{\hat{A},\hat{\B}'}  =[\id_{\FAout \rightarrow \Aout} \otimes   \Lambda^2_{\aux \FBout \rightarrow \Bcheatout}]  \circ [\mF_{\FAin \FBin \rightarrow \FAout \FBout} \otimes \id_{K}]  \circ [\id_{\Ain \rightarrow \FAin} \otimes \Lambda^1_{\Bin \rightarrow \FBin \aux}],$
where $\circ$ denotes sequential application of CPTP maps.

In the following we let $\eps \geq 0$ and write $\rho \approx_\eps
\sigma$ if $C(\rho, \sigma) \leq \eps$. $C(\rho, \sigma)$ is the
purified distance, defined as $\sqrt{1-F(\rho, \sigma)^2}$ for
$F(\rho, \sigma):=\tr \sqrt{\sqrt{\rho}\sigma \sqrt{\rho}}$ the
fidelity. 

\begin{definition} A (two-party quantum) protocol $\qp$ for $f$ is {\em $\eps$-correct} if 
for any distribution $p(u,v)$ of the inputs it holds that
\begin{align*} [\id_{\refe} \otimes & \qp_{\A,\B}] (\rho_{\Ain \Bin \refe})
\epsclose[\eps] [\id_{\refe} \otimes  \mF_{\hat\A,\hat{\B}}]
(\rho_{\Ain \Bin \refe}).
\end{align*}
The protocol is \emph{ $\eps$-secure against dishonest Bob} if for
any $p(u,v)$ and for any real adversary $\dB$, there exists an ideal
adversary $\dhB$ such that 
\begin{align*}[\id_{\refe} \otimes  &\qp_{\A,\B'}] (\rho_{\Ain \Bin
    \refe}) \epsclose[\eps] [\id_{\refe} \otimes  \mF_{\hat\A,\dhB}]
  (\rho_{\Ain \Bin \refe}).
\end{align*}
\emph{$\eps$-security against dishonest Alice} is defined
analogously. 
\end{definition}

Since $\mF$ is classical, we can augment it so that it outputs $\tilde{v}$ in addition. More precisely, we define $\mFaug: \FAin\FBin \rightarrow \FAout \FBout \FBin $ by
$\mF_{aug}(\ketbra{u}{u'}_{\FAin} \otimes
\ketbra{v}{v'}_{\FBin}) :=\delta_{u,u'} \delta_{v,v'} \proj{f(u, v)}_{\FAout} 
 \otimes \proj{f(u, v)}_{\FBout} \otimes \proj{v}_{\FBin}. 
$
which has the property that $ \mF=\tr_{\FBin} \circ  \mFaug.$
For a concrete input distribution we define $\sigma_{\refe \Aout \FBin \Bcheatout}:= [\id_{\refe} \otimes  \mF_{\hat{A},\hat{\B}'\!,\mathit{aug}}] (\rho_{\Ain\Bin\refe})$
which satisfies
  $\sigma_{RXY'}\approx_\eps \rho_{RXY'}$
for $\rho_{\refe \Aout \Bcheatout}:=[\id_{\refe} \otimes  \qp_{\A,\B'}] (\rho_{\Ain\Bin\refe}) $ if the protocol is secure against cheating Bob.
We call $\sigma_{\refe \Aout \FBin \Bcheatout}$ a \emph{secure state
  for input distribution $p(u, v)$}. 

\medskip\noindent
{\bf Main Results.}
The proofs of our main results build upon the following lemma which
constructs a cheating strategy for Alice that works \emph{on average}
over the input distribution $p(u, v)$. Subsequently we will show how this result can be used to devise a cheating strategy that works $\emph{for all}$ distributions at the same time. 

\begin{lemma}\label{lem:robust}
If a protocol $\pi$ for the evaluation of $f$ is $\eps$-correct and $\eps$-secure against Bob, then for all input distributions $p(u, v)$ there is a cheating strategy for Alice such that she obtains $\tilde{v}$ with some probability distribution $q(\tilde{v}|u,v)$ satisfying
$\sum_{u, v, \tilde{v}} p(u, v) q(\tilde{v}|u,v) \delta_{f(u, v),
  f(u, \tilde{v})}\geq 1-6\eps \, .
$ Furthermore, $q(\tilde{v}|u,v)$ is almost independent of $u$; i.e., there exists a distribution  $\tilde{q}(\tilde{v}|v)$ such that 
$\sum_{u, v, \tilde{v}} p(u, v) |q(\tilde{v}|u,v)
-\tilde{q}(\tilde{v}|v)|   \leq 6\eps \, .
$
\end{lemma}

\begin{proof} We first construct a ``cheating unitary'' $T$ for Alice and then show how Alice can use it to cheat successfully.

Let Alice and Bob play honestly but let them purify their protocol with purifying registers $\Apur$ and $\Bpur$ respectively. We assume without loss of generality that honest parties measure their classical input and hence, $\Apur$ and $\Bpur$ contain copies of $u$ and $v$, respectively. We denote by $\ket{\Phi}_{\refe \Aout \Apur \Bpur \Bout}$ the state of all registers at the end of the protocol. Notice that tracing out $\Apur$ from $\ket{\Phi}_{\refe \Aout \Apur \Bpur \Bout}$ results in a state $\tr_{\Apur}{\proj{\Phi}}_{\refe \Aout \Apur \Bpur \Bout}=\rho_{\refe\Aout \Bpur \Bout}$ which is exactly the final state when Alice played honestly and Bob played dishonestly with the following strategy: he plays the honest but purified strategy and outputs the purification of the protocol (register $\Bpur$) and the output values $f(u,v)$ (register $\Bout$). His combined dishonest register is $\Bcheatout=\Bpur\Bout$. Since the protocol is $\eps$-secure against Bob by assumption, there exists a secure state $\sigma_{\refe\Aout \FBin \Bcheatout}$ satisfying
\begin{equation} \label{eq:sec}
\sigma_{\refe\Aout\Bcheatout }\approx_\eps \rho_{\refe\Aout\Bcheatout }.
\end{equation} 
Let $\ket{\Psi}_{\refe \Aout \pur \FBin \Bcheatout}$ be a purification
of $\sigma_{\refe \Aout \FBin \Bcheatout}$ with purifying register
$\pur$. Note that $\ket{\Psi}_{\refe \Aout \pur \FBin \Bcheatout}$ is
also a purification of $\sigma_{\refe\Aout\Bcheatout}$, this time with
purifying registers $\pur \FBin $. Recall that $\ket{\Phi}_{\refe \Aout\Apur \Bcheatout}$ purifies
$\rho_{\refe \Aout\Bcheatout}$ with purifying register $\Apur$.  Since
 $\sigma_{\refe\Aout\Bcheatout }\approx_\eps \rho_{\refe\Aout\Bcheatout }$ we can use Uhlmann's theorem~\cite{Uhlmann} to conclude that there
exists an isometry $T\equiv T_{\Apur \rightarrow \pur \FBin}$ (with induced CPTP map ${\cal T} \equiv {\cal T}_{\Apur \rightarrow \pur \FBin}$) such that
\begin{equation} \label{eq:unitary} [{\cal T}_{X'_1 \rightarrow P \tilde{V}} \otimes
  \id_{\refe \Aout \Bcheatout}] (\proj{\Phi}_{\refe \Aout\Apur \Bcheatout} )\approx_\eps \proj{\Psi}_{\refe\Aout\pur \FBin \Bcheatout} \, .
\end{equation}
This concludes the construction of $T\equiv T_{\Apur \rightarrow \pur \FBin}$.

We will now show how Alice can use the isometry $T$ to cheat.  Notice that
tracing out $\Bpur$ from $\ket{\Phi}_{\refe \Aout\Apur \Bout\Bpur}$
results exactly in the final state when Bob played honestly and Alice
played dishonestly with the following strategy: she plays the honest
but purified strategy and outputs the purification of the protocol
(register $\Apur$) and the output values $f(u,\hut{v})$ (register
$\Aout$). She then applies $T_{\Apur \rightarrow \pur \FBin}$, measures register $\FBin$ in the computational basis and obtains a
value $\tilde{v}$. It remains to argue that Alice can compute $f(\hut{u},\hut{v})$ with good probability based on
the value $\tilde{v}$ that she has obtained from measuring register
$\FBin$.

Let $\mathcal{M}_{\refe \FBin \Aout}$ be the CPTP map that measures registers
$\refe,\Aout$ and $\FBin$ in the computational basis. Tracing over
$\pur\Bcheatout$ and applying $\mathcal{M}_{\refe \FBin \Aout}$ on both sides of
Equation~\eqref{eq:unitary}, we
find
\begin{align}
[\mathcal{M}_{\refe \Aout\FBin }& \otimes \tr_{\pur\Bcheatout}] ([{\cal T}_{\Apur \rightarrow \pur \FBin} \otimes
  \id_{\refe \Aout \Bcheatout}](\proj{\Phi}_{\refe \Aout\Apur \Bcheatout} )) \nonumber \\ 
  & \approx_\eps [\mathcal{M}_{\refe  \Aout\FBin} \otimes  \tr_{\pur\Bcheatout} ](\proj{\Psi}_{\refe\Aout\pur \FBin \Bcheatout}) \, \label{eq:5}
\end{align} 
by the monotonicity of the purified distance under CPTP maps. The
right-hand side of Equation~\eqref{eq:5} equals 
\begin{align*}
\sum_{u,v, \tilde{v}} p(u,v) \tilde{q}(\tilde{v}|\hut{v})\proj{u v}_R
\otimes \proj{\tilde{v}}_{\FBin} \otimes 
\proj{f(\hut{u},\tilde{v})}_{\Aout} \, ,
\end{align*}
for some probability distribution $\tilde{q}(\tilde{v}|v)$ that is conditioned only on
  Bob's input $v$, since $\ket{\Psi}_{\refe\Aout\pur \FBin \Bcheatout}$ is a purification of the secure state
  $\sigma_{\refe\Aout \FBin \Bcheatout}$. The left-hand side
  of Equation~\eqref{eq:5} equals
\begin{align} 
\begin{split}
\sum_{u,v, \tilde{v}, x}  p(u,v) q(\tilde{v}|u, v)  & \proj{u
  v}_{\refe} \otimes  \proj{\tilde{v}}_{\FBin} \\
& \otimes r(x|u, v, \tilde{v})  \proj{x}_{\Aout} 
\end{split} \label{eq:flight0}
\end{align}
for some conditional probability distributions $q(\tilde{v}|u,
\hut{v})$ and $ r(x|u, v, \tilde{v})$. Because of
the correctness of the protocol, term~\eqref{eq:flight0} is $\eps$-close to 
\begin{equation}\label{eq:flight}
\sum_{u,v, \tilde{v}} p(u,v) \bar{q}(\tilde{v}|u, v)   \proj{u
  v}_{\refe}  \otimes \proj{\tilde{v}}_{\FBin} \otimes \proj{f(\hut{u},\hut{v})}_{\Aout}
\end{equation}
for some conditional probability distribution $ \bar{q}(\tilde{v}|u,
v) $.  Noting that the $\eps$-closeness of \eqref{eq:flight0} and
\eqref{eq:flight} implies that $p(\cdot, \cdot) q(\cdot|\cdot, \cdot)$
and $p(\cdot, \cdot) \bar{q}(\cdot|\cdot, \cdot)$ (when interpreted as quantum states) are $\eps$-close in
purified distance, we can replace $p(\cdot, \cdot)
\bar{q}(\cdot|\cdot, \cdot)$ in \eqref{eq:flight} by $p(\cdot, \cdot)
q(\cdot|\cdot, \cdot)$ increasing the purified distance to the
left-hand side of Equation~\eqref{eq:5} only to $2\eps$. Putting things together,
Equation~\eqref{eq:5} implies
\begin{align}  \label{eq:sameoutput}
&\sum_{u,v,\tilde{v}} p(u,v) q(\tilde{v}|u, v)  \proj{u v}_{\refe} \,
\proj{\tilde{v}}_{\FBin} \,
\proj{f(\hut{u},\hut{v})}_{\Aout} \\    \nonumber
&\approx_{3\eps} \sum_{u,v, \tilde{v}} p(u,v)
\tilde{q}(\tilde{v}|\hut{v}) \proj{u v}_R \, \proj{\tilde{v}}_{\FBin}
\, \proj{f(\hut{u},\tilde{v})}_{\Aout}   \, .
 \end{align}

Sandwiching both sides with $\tr [Z \cdot ]$, where 
$Z=\sum_{u,v, \tilde{v}} \proj{u v}_R \otimes \proj{\tilde{v}}_{\FBin}
\otimes \proj{f(\hut{u},\tilde{v})}_{\Aout} $
we find the first claim since the purified distance of two distributions upper bounds their total variation distance and since the latter does not increase under $\tr [Z \cdot ]$. The second claim follows similarly by
tracing out register $\Aout$ from Equation~\eqref{eq:sameoutput}. 
\end{proof}

Applying the lemma to the uniform distribution we immediately obtain our impossibility result for perfectly secure protocols.

\begin{theorem} \label{thm:perfect} Let $\qp$ be a protocol for the
  evaluation of $f$ which is
 perfectly correct and perfectly secure ($\eps=0$) against Bob. Then, if Bob has input $v$, Alice can compute $f(u,v)$ for all $u$.
\end{theorem}
We note that this implies that Alice can
completely break the security for any non-trivial function $f$.
\begin{proof}
Letting $p(u, v)=\frac{1}{|U||V|}$ and $\eps=0$ in the lemma results in the statement that if Alice has input $u_0$, then she will obtain $\tilde{v}$ from the distribution $q(\tilde{v}|u_0, v)$ which equals $\tilde{q}(\tilde{v}|v)$. But since also $q(\tilde{v}|u, v)=\tilde{q}(\tilde{v}|v)$ for all $u$, we have 
$\frac{1}{|U||V|}\sum_{u, v, \tilde{v}}  q(\tilde{v}|u_0, v) \delta_{f(u, v), f(u, \tilde{v})}=1. $
In other words, all $\tilde{v}$ that occur (i.e.~that have $\tilde{q}(\tilde{v}|v)>0$) satisfy for all $u$, $f(u, v)=f(u, \tilde{v})$. Alice can therefore compute the function for all $u$.
\end{proof}

The impossibility result for the case of imperfect protocols is also based on the lemma, but requires a subtle swap in the order of quantifiers (from ``$\forall$ input $\exists$ ideal adversary'' to ``$\exists$ ideal adversary $\forall$ input'') which we achieve by use of von Neumann's minimax theorem. 

\begin{theorem} \label{th:minimax} If a protocol $\qp$ for the evaluation of $f$ is $\eps$-correct and $\eps$-secure
  against Bob, then there is a cheating strategy for Alice (where she
  uses input $u_0$ while Bob has input $v$) which gives her
  $\tilde{v}$ distributed according to some distribution
  $Q(\tilde{v}|u_0, v)$ such that for all $u$:
$\Pr_{\tilde{v} \sim Q}[f(u,v) = f(u,\tilde{v})] \geq 1- 28 \eps$.
\end{theorem}
\begin{proof}
%

  The argument is inspired
  by~\cite{KretschmannWernerBitCommit2007}. For a finite set $\mathcal{S}$, we
  denote by $\Delta(\mathcal{S})$ the simplex of probability distributions
  over $\mathcal{S}$. Denote by $\cW$ the set of
  pairs $(u, v)$. Consider a finite $\eps$-net $\cD$ of
  $\Delta(\cW)$ in total variation distance; and to each distribution
  in $\cD$ the corresponding cheating unitary $T$ constructed in the
  proof of the lemma. We collect all these unitaries in the (finite) set $\cE$
  and assume that $T$ determines $p$ uniquely, as we could include the
  value $p$ into $T$. For each such $T$, let $q(\tilde{v}|u,v, T) $
  and $\tilde{q}(\tilde{v}|v, T)$ be the distributions from the
  lemma. Define the payoff function
$g(u, v, T)  := \sum_{\tilde{v}} q(\tilde{v}|u,v, T) \delta_{f(u, v),f(u, \tilde{v})}  -\sum_{\tilde{v}} |q(\tilde{v}|u,v, T) -\tilde{q}(\tilde{v}|v, T)| $.
The lemma then yields $1-12\eps \leq \min_{p\in \cD} \max_{T\in \cE} \sum_{u,v} p(u,v) g(u,
v, T)$ which is at most $2\eps + \min_{p'\in \Delta(\cW)} \max_{T\in \cE} \sum_{u,v}
p'(u,v) g(u, v, T)$, since replacing $p$ by $p'$ incurs
only an overall change in the value by $2\eps$ (as $-1 \leq
g(u,v,T)\leq 1$) . By von Neumann's minimax theorem, this last term equals 
$2\eps + \max_{p''\in \Delta(\cE)} \min_{(u, v) \in \cW} \sum_{T}
g(u, v, T)p''(T)$ \footnote{In order to apply von Neumann's
  theorem, note that the initial term equals
  $\min_{p'\in \Delta(\cW)} \max_{p''\in \Delta(\cE)} \sum_{u,v}
  p'(u,v) g(u, v, T) p''(T)$ since the maximal value of the latter is
  attained at an extreme point. Von Neumann's minimax theorem
  \cite{Neumann-Gesellschaftsspiele} allows us to swap minimization and
  maximization leading to $\max_{p''\in \Delta(\cE)} \min_{p\in
    \Delta(\cW)} \sum_{u, v, T} p(u, v) g(u, v, T)p''(T)$ without
  changing the value. This expression corresponds to the final term since the minimization can without loss of
  generality be restricted to its extreme points .}.



Hence, we have shown 
that there is a strategy for Alice, where she chooses her cheating unitary $T$ with probability $p''(T)$, such that (for some $\eps_1+\eps_2\leq 14\eps$) for all $u, v$,
\begin{equation} \label{eq:correct} 
\sum_{\tilde{v}}  Q(\tilde{v}|u,v) \delta_{f(u, v),f(u, \tilde{v})}\geq 1-\eps_1
\end{equation}
and
$  \sum_{\tilde{v}}   |Q(\tilde{v}|u,v) -\tilde{Q}(\tilde{v}|v)| \leq \sum_{\tilde{v}, T} p(T) |q(\tilde{v}|u,v, T) -\tilde{q}(\tilde{v}|v, T)|\leq \eps_2,$
where $Q(\tilde{v}|u,v) := \sum_{T} p(T)  q(\tilde{v}|u,v,
T) $ and $\tilde{Q}(\tilde{v}|v):=  \sum_{T} p(T)
\tilde{q}(\tilde{v}|v, T)$. This implies that for all $u, v$,
$\sum_{\tilde{v}}| Q(\tilde{v}|u_0,v)- Q(\tilde{v}|u,v)|\leq
2\eps_2 \, $.
Combining this inequality with Equation~\eqref{eq:correct}, we find for all $u, v$,
$\sum_{\tilde{v}}   Q(\tilde{v}|u_0,v) \delta_{f(u, v),f(u, \tilde{v})}\geq 1-\eps_1-2\eps_2 \geq 1-28\eps \, .$
\end{proof}


One might wonder whether Theorem~\ref{th:minimax} can be strengthened
to obtain, with probability $1-O(\eps)$,  a $\tilde{v}$  such that for
all $u: f(u,v)=f(u,\tilde{v})$. It turns  out that this depends on the
function $f$: when $f$ is equality $[\text{EQ}(u,v) = 1$ iff $u=v]$
and inner-product modulo 2 $[\text{IP}(u,v)=\sum_i u_i\cdot v_i \mod 2]$,  the stronger conclusion is possible. However for disjointness $[\text{DISJ}(u,v) =0$ iff $\exists i: u_i=v_i=1]$ such a strengthening is not possible showing that our result is tight in general.

For EQ, we reason as follows. Set $u=v$ in Theorem~\ref{th:minimax}. Alice is able to sample a $\tilde{v}$ such that $\sum_{\tilde{v}} Q(\tilde{v}|u_0,v) \delta_{EQ(v, v),EQ(v, \tilde{v})} \geq 1-28\eps$. Since $\delta_{\text{EQ}(v, v),\text{EQ}(v, \tilde{v})}=1$ iff $v=\tilde{v}$, $Q(v|u_0,v) \geq 1-28\eps$.
When $f$ is IP, we pick $u$ uniform at random and obtain $\sum_{\tilde{v}} Q(\tilde{v}|u_0,v)(2^{-n}\sum_{u}  \delta_{\text{IP}(u,v),\text{IP}(u,\tilde{v})})\geq 1-28\eps$. Using $2^{-n}\sum_{u}  \delta_{\text{IP}(u,v),\text{IP}(u,\tilde{v})} =1$ if $\tilde{v}=v$, and $\frac{1}{2}$ if $\tilde{v}\neq v$, we find $Q(v|u_0,v) + \frac{1}{2}(1-Q(v|u_0,v)) \geq 1-28\eps$, which implies $Q(v|u_0,v) \geq 1-56\eps$.
Interestingly, for DISJ such an argument is not possible. Assume that
we have a protocol  that is  $\eps$-secure against Bob. Bob could now
run the protocol normally on strings $v$ with Hamming weight $|v| \leq
n/2$, but on inputs $v$ with $|v|>n/2$ he could flip, at random,
$\sqrt n$ of $v$'s bits that are 1. It is not hard to see that this
new protocol is still $\eps$-secure and $\eps +
O(\frac{1}{\sqrt{n}})$-correct. The loss in the correctness is due to
the fact that, on high Hamming-weight strings, the protocol may, with
a small probability, not be correct. On the other hand, on
high-Hamming-weight inputs, the protocol can not transmit or leak the complete input $v$ to Alice, simply because Bob does not use it.


\medskip\noindent {\bf Acknowledgments.}  We thank Anne Broadbent,
Ivan Damg{\aa}rd, Fr\'ed\'eric Dupuis, Louis Salvail, Christopher
Portmann and Renato Renner for valuable discussions, and an anonymous
referee for suggesting an example presented in the appendix. M.C.~is
supported by the Swiss National Science Foundation (Grant
No. PP00P2-128455 and 20CH21-138799), the NCCR ``Quantum Science and
Technology,'' and the German Science Foundation (Grant No.~CH
843/2-1). C.S. is supported by a NWO Veni grant. H.B. was supported by
an NWO VICI grant and by EU project QCS.

\appendix*

\section{Appendix: Additional Comments about the Security Definition} \label{appendix}
Since this work presents impossibility results for the secure
computation of $f$, one may wonder how the results are affected when the notions of security are weakened. In particular, one may ask whether similar results can be obtained when, instead of the real/ideal-world paradigm, notions of security more akin to the ones used in the well-known no-go proofs for bit commitment and one-sided computation would be used. Whereas we do not know the answer to this
question in general, we wish to emphasize the difficulty in
formalizing such notions of security satisfactorily. 

With regards to the real/ideal-world paradigm we will now comment on some specific notions of security used in this work. A central object in the real/ideal-world paradigm is the ideal functionality. Since we are faced with the task of the secure evaluation of a
\emph{classical} deterministic function, we chose to consider an ideal
functionality which measures the inputs it receives and outputs
orthogonal states to the parties that correspond to the function
values. Note that in certain situations one may be satisfied with
different (possibly weaker) ideal functionalities for this task; we leave open the question to what extend our results remain valid in such situations.

One may also wonder whether the
purification of the inputs could not be omitted. Note that such an omission would
correspond to a serious limitation of the environment to distinguish
the real and from the ideal world. With respect to the stronger notion
of security discussed in the main text, for instance, there can be a large
difference between the diamond norm (which corresponds to purified
inputs) and the induced norm (where the maximisation is over inputs
that are not purified), see e.g.~\cite{KretschmannWernerBitCommit2007}. 
This difference does not
occur in the case of perfectly secure protocols, where one can
therefore omit the reference. The omission of the reference has a more
serious effect on the weaker notion of security considered in this work, even
in the case of perfect security, since we only consider (purified)
classical inputs; in fact, omission would invalidate the no-go result as we will now show. We leave it as an open question whether Theorem 2 can be proven were arbitrary (unpurified) inputs considered.


The following example was suggested to us by an anonymous
referee and shows the necessity of requiring the register $R$ in our security definition. Consider the classical deterministic function
$f((s_0,s_1),b)=(b,s_b)$ of $n$-bit strings $s_0,s_1$ and a choice bit
$b$ which is inspired by a one-out-of-two-string-oblivious transfer
but outputs both the choice bit and the string of choice to both Alice
and Bob. Let us consider the following protocol $\pi_{\A,\B}$: Bob
sends $b$ to Alice and Alice responds with $s_b$.

Clearly, this protocol is secure against cheating Bob, who learns no more than either $s_0$ or $s_1$.
One might also think that this protocol is perfectly secure
against cheating Alice because Alice learns Bob's choice bit
anyway. Indeed, if we defined security without purifying register $R$ one
could construct an ideal adversary Alice $\dhA$ from any real
adversary $\dA$ as follows. Let $\dhA$ simulate two independent copies
of $\dA$ and give $b=0$ to the first and $b=1$ to the second copy
which both respond with a string $s_0$ and $s_1$, respectively. Let
$\dhA$ input these two strings $(s_0,s_1)$ into the ideal
functionality $\mF$ and receive $(b,s_b)$ as output from $\mF$. Output
whatever the real copy of $\dA$ corresponding to the bit $b$ outputs
(and discard the other copy). This simulation generates an output in
the ideal world which is identically distributed to the one from the
real protocol. Hence, the protocol would be perfectly secure against Alice. Notice
that this example shows that an analogue of our Theorem~1 cannot be
proven for this weaker security definition.

We stress that the above protocol is \emph{not} secure according to
our security definition by virtue of the purifying register $R$. Consider
the uniform input distribution over $n$-bit strings $(s_0,s_1)$ in the
$2n$-qubit register $U$ and the choice bit $b$ in register $V$. Hence,
the input state $\rho_{RUV}$ if fully entangled between $R$ and
$UV$. Let us consider the following real adversary $\dA$ who measures
the first $n$ qubits of $U$ in the computational basis in case $b=0$
or performs the measurement in the Hadamard basis if $b=1$ and returns
the measurement outcome as $s_b$. Due to the entanglement, the first
$n$ qubits of $R$ collapse to the measured state. Notice that
for this adversary $\dA$, the argument above is no longer applicable,
because $\dhA$ cannot simulate two independent copies of $\dA$ as the
$U$ register is only available once. In fact, for this adversarial
strategy $\dA$, only one of the two strings $s_0, s_1$ is well-defined
as the other string corresponds to the measurement outcome in a
complementary basis of the same quantum state. This highlights the
intuitive security problem of the suggested protocol, namely that it
is not guaranteed that both $s_0$ and $s_1$ classically exist for a
cheating Alice. This shows that the protocol is not secure against cheating Alice and that it therefore does not stand in contradiction with our results.

\bibliographystyle{alpha}
\bibliography{impossible}

\end{document}